\documentclass{IEEEtran}

\usepackage[nospace]{cite}

\usepackage{mathtools, amssymb, fontawesome5}
    \def\A{\mathsf A}
    \def\C{\mathsf C}
    \def\G{\mathsf G}
    \def\T{\mathsf T}
    \def\le{\leqslant}
    \def\ge{\geqslant}
    \def\EE{\mathbb E}
    \def\PP{\mathbb P}

    \newtheorem{theorem}{Theorem}
    \newtheorem{lemma}[theorem]{Lemma}
    \newtheorem{proposition}[theorem]{Proposition}
    \newtheorem{remark}[theorem]{Remark}
    \def\bca#1{\begin{cases*} #1 \end{cases*}}
    \def\bma#1{\begin{bmatrix} #1 \end{bmatrix}}
    \def\sma#1{\left[ \begin{smallmatrix} #1 \end{smallmatrix} \right]}

\usepackage{tikz, booktabs}
    \tikzset{every picture/.style = {line cap=round,line join=round}}

\begin{document}

                                   \title
                   {Frequency Coding over Noisy Sampling}
                                      
                                   \author
               {Bo-Yu Su, Hsin-Po Wang, Venkatesan Guruswami}

                                 \maketitle

\begin{abstract} \boldmath
    DNA molecules are so small that it might be practical to use their
    frequency vectors to encode messages.  More precisely, a sender can
    inject $M_X$ copies of the string $X =$ CATCATCAT into a pool and the
    receiver can recover $M_X$ by sequencing the pool.  There are, however,
    two sources of uncertainty: (a) $M_X$ is usually too big to be counted
    exactly, but is estimated by sampling.  (b) The DNA sequencer could be
    noisy; it may have difficulty distinguishing CATCATCAT from CATGATCAT.
    
    Recently, Tamir, Weinberger, and Guillén i Fàbregas clarified the amount
    of information the frequency vector can carry under (a).  They showed
    that each string can carry about $\log_4 R$ bits, where $R$ is the
    average number of times each string is read.  They also showed that
    $\log_4 R$ bits can be achieved by a low-complexity uncoded scheme under
    the condition that there are at least $\sqrt R$ distinct strings.  In
    this paper, we show that a low-complexity coded scheme can achieve the
    same $\log_4 R$ bits unconditionally.  We then generalize the scheme to
    handle sequencing noise, (b), and show that the noise penalizes the total
    number of bits by $\log_2 \det W$, together with a linear term due to the
    use of Fourier transforms in our proof.  The former penalty $\log_2 \det
    W$ is asymptotically the same as that obtained by Gerzon, Shomorony, and
    Weinberger; our scheme trades a small amount of rate for practical
    complexity.
\end{abstract}

\footnotetext[0\def\thefootnote{\faEnvelope[regular]}]{
    Bo-Yu Su and Hsin-Po Wang are with National Taiwan University;
    emails: \{R14942154, hsinpo\}@ntu.edu.tw.
    Venkatesan Guruswami is with the University of California, Berkeley;
    some of this work was done when the author was visiting
    the National Taiwan University; email: venkatg@berkeley.edu.
}

\footnotetext[0\def\thefootnote{\faDollarSign}]{
    This work was supported in part by
    the National Science and Technology Council (NSTC), Taiwan,
    under Grant NSTC 114-2222-E-002-007.
}

% \footnotetext[0\def\thefootnote{\faChalkboardTeacher}]{}

\section{Introduction}

    DNA is an interesting candidate medium for data storage with
    immense potential.  It is dense, durable, and easy to replicate
    \cite{MiP24,SKS24,XLC24}.  However, DNA data storage poses unique
    challenges due to the biochemical processes involved in writing and
    reading data.  One challenge is that we cannot synthesize arbitrarily
    long DNA strands with high fidelity.  And even if we try to,
    a long strand tends to break into shorter fragments over time
    \cite{BMY22,RVS21,ShV21}.  Therefore, current studies on DNA data storage
    systems focus on using a pool of multiple short strands \cite{RGC26} to
    represent the data collectively.

    There are two regimes.  The long strand regime is when
    \[ M \ll 4^\ell \]
    where $M$ is the number of strands and $\ell$ is the length of each
    strand.  It can be shown that the maximum amount of information each
    strand can carry is about $(\ell - \log_4 M) \log_2 4$ bits.  This
    quantity can be explained by prefixing each strand with an index about
    $\log_4 M$ nucleotides long, and using the remaining $\ell - \log_4 M$
    nucleotides to carry information.  When there is noise, the story becomes
    more involved.  Readers are referred to \cite{WeM22,LSW23,WaG25} for the
    latest updates.
    
    The short strand regime is when
    \[ M \gg 4^\ell. \]
    Such $\ell$ is not long enough for unique indices, let alone the payload.
    Hence, in this regime, each \emph{logical} string $X \in \{\A, \C, \G,
    \T\}^\ell$ will appear in multiple \emph{physical} strands.  The main
    carrier of information is the number of strands that share the same
    string.  That is, we use the frequency vector
    \begin{equation}
        F\colon \{\A, \C, \G, \T\}^\ell \to [0, 1]                  \label{F}
    \end{equation}
    to represent data, where $F(X) M$ is the number of strands that carry
    string $X$, and $\sum_X F(X) = 1$.

    A natural question is how much information each frequency $F(X) \in [0,
    1]$ can carry.  It is unreasonable to expect that we can learn this real
    number with infinite precision and obtain infinite information.  Instead,
    a widely accepted model posits that the receiver can only sample the DNA
    pool $N \approx M$ times, and hope that each string will appear as close
    to $F(X) N$ times as possible.  Since $N$ is usually huge in practice,
    the actual number of reads that result in $X$ can be modeled by a Poisson
    random variable with mean and variance both $F(X) N$.  This means that
    the uncertainty of estimating $F(X) N$ is within
    \begin{equation}
        \pm \Theta(\sqrt{F(X) N})                                \label{sqrt}
    \end{equation}
    with high probability.

    For now, let us assume that every $F(X)$ is in $(1 \pm 1/2) / 4^\ell$.
    Then $F(X) N$ lies in the range $\Theta(N / 4^\ell)$.  With uncertainty
    $\Theta(\sqrt{F(X) N}) = \Theta(\sqrt{N / 4^\ell})$, the range of $F(X)
    N$ can fit $\Theta (\sqrt{N / 4^\ell})$ distinguishable quantization
    levels (QLs).\footnote{ A quantization level is one discrete
    representative value that the encoder may assign to $F(X) N$; two levels
    are distinguishable when their separation is larger than the sampling
    uncertainty.} This many levels can carry $\log_2 \sqrt{N / 4^\ell} =
    \log_4 N - \ell$ bits of information.  Note that $N / 4^\ell$, denoted by
    $R$ hereafter, is also the average number of times each string is read.
    Therefore, one can expect that each string carries about $\log_4 R$ bits.

    To what extent does the expectation hold?  The quantity $\log_4 R$ bits
    per string was first conjectured in \cite{ShH22}.  Later, in
    \cite{GSW25}, Gerzon, Shomorony, and Weinberger proved that $\log_4 R$ is
    achievable and order-optimal for the range $4^\ell < M < 4^{2\ell}$.
    (Since it is assumed that $M \propto N$, this is equivalent to saying
    $4^\ell < N < 4^{2\ell}$, ignoring constant factors.) Recently, Tamir,
    Weinberger, and Guillén i Fàbregas \cite{TWG25} showed that the $\log_4
    R$ bound is achievable and order-optimal for the entire short-strand
    regime $4^\ell < M$.  They also proposed a low-complexity uncoded scheme
    that achieves $\log_4 R$ when $4^\ell < M < 4^{3\ell}$.
    
    Now the only question that remains is whether we can construct a
    low-complexity scheme achieving $\log_4 R$ bits per string for the
    very-short regime $4^{3\ell} < M$.  In this paper, we provide exactly
    that.

    \begin{theorem}[main theorem, noiseless case]       \label{thm:noiseless}
        For the entire short-strand regime, there exists a low-complexity
        scheme that encodes $(1 - o(1)) \log_4 R$ bits per string, or $(1 -
        o(1)) 4^\ell \log_4 R$ bits per pool.  When the gap to capacity is
        fixed, the error probability decays exponentially in $4^\ell$.
    \end{theorem}

    Theorem~\ref{thm:noiseless} can be generalized to the case with
    sequencing noise.  Consider the model where the DNA sequencer is a
    quaternary symmetric channel with substitution probability $p < 3/4$:
    \begin{equation}
        V(y|x) = \bma{
            1-p & p/3 & p/3 & p/3 \\
            p/3 & 1-p & p/3 & p/3 \\
            p/3 & p/3 & 1-p & p/3 \\
            p/3 & p/3 & p/3 & 1-p
        }                                                     \label{channel}
    \end{equation}
    The noise morphs $F$ into a less informative distribution $WF$, where $W
    \coloneqq V^{\otimes\ell}$ is a matrix describing the noise on the string
    level and $F$ is treated as a column vector.  As $F$ lives in the
    probability simplex
    \[
        \Bigl\{x : \sum x = 1 \text{ and } x \ge 0 \Bigr\}
        \subset [0, 1]^{4^\ell},
    \]
    multiplying by $W$ reduces its volume by a factor of $\det W$.  Hence,
    assuming the same quantization density, the number of distinguishable
    levels is reduced by a factor of $\det W$, and the number of bits they
    carry is penalized by $\log_2 \det W = 3 \ell 4^\ell \log_{16} \delta$,
    where $\delta \coloneqq 1 - 4p/3$.

    \begin{theorem} [main theorem, noisy case]              \label{thm:noisy}
        For the very-short-strand regime where $r \coloneqq \log_4 R - \ell +
        3\ell \log_{16} \delta$ is positive, there exists a low-complexity
        scheme that encodes $(1 - o(1)) r$ bits per string, or $(1 - o(1))
        (4^\ell (\log_4 R - \ell) + 3 \ell 4^\ell \log_{16} \delta)$ bits per
        pool.
    \end{theorem}

    To summarize, Theorem~\ref{thm:noiseless} achieves the capacity proved by
    \cite{TWG25}, completing the noiseless picture.  Theorem~\ref{thm:noisy}
    \emph{almost} achieves the $\log_2 \det W$ penalty that is reported in
    Gerzon, Shomorony, and Weinberger \cite{GSW26}.  The extra $-\ell$ term
    is due to the use of Fourier transforms in our proof.  Gerzon, Shomorony,
    and Weinberger \cite{GSW26} also consider more general channels, while we
    only focus on $V$ for a proof of concept.

    The paper is organized as follows.  Section~\ref{sec:back}
    reviews the current state of DNA technology and states the problem.
    Section~\ref{sec:uncoded} describes the uncoded scheme proposed
    by Tamir, Weinberger, and Guillén i Fàbregas \cite{TWG25}.  In
    Section~\ref{sec:noiseless}, we present our new coding scheme for
    noiseless sampling.  In Section~\ref{sec:noise}, we generalize it
    to the noisy sampling case.

\section{DNA Background}                                     \label{sec:back}

    Storing data on DNA is a promising technology that has been discussed and
    researched for over two decades.  DNA molecules are exceptionally stable
    compared to artificial media of the same size.  But that does not mean it
    is error-free.  In this section, we briefly explain current
    biotechnologies for writing, copying, and reading data, and derive a
    mathematical model.

\subsection{Synthesis (Writing)}

    Synthesizing a new DNA strand is done on a silicon-based chip analogous
    to a typewriter: There is a keyboard with four keys \fbox{$\A$},
    \fbox{$\C$}, \fbox{$\G$}, and \fbox{$\T$}.  Pressing each key will append
    the corresponding nucleotide to the tail of the current working strand.
    On the downside, the typewriter is slow, producing one nucleotide per
    minute \cite{FKK25}.  On the upside, we can increase the number of
    parallel typewriters until the overall data throughput becomes useful.

    The typewriters are not only slow but error-prone \cite{FMS21,MOS22}; the
    longer the strand, the higher the chance of typos.  On the other hand,
    there is not a clear limit on the number of typewriters one can use in
    parallel \cite{YTL24}.  Accordingly, synthesized DNA strands will almost
    always be very short (about hundreds of nucleotides long) compared to the
    number of typewriters (potentially in the millions).

\subsection{Amplification (Copying)}

    One feature that sets DNA apart from traditional data storage media is
    that replicating DNA is nearly effortless.  Unlike synthesizing new
    strands, which requires silicon chips, copying existing strands is done
    using specialized enzymes called \emph{polymerases}.  Commercially
    available polymerases can duplicate DNA at a rate of 1000 nucleotides per
    second, and they work in parallel.  Even if we account for the time for
    thermal cycling, the number of DNA strands \emph{doubles} every few
    minutes \cite{Sai88}.

    Another feature of enzyme-based amplification is its fidelity.
    Polymerases are responsible for copying DNA in living organisms, making
    $1$ mistake every $10{,}000$ nucleotides \cite{Joh93}.  This is
    effectively error-free considering that other steps (writing and reading)
    introduce significantly more errors.  If anything, we only need to worry
    about amplification bias.  That is, in every thermal cycle, the fraction
    of strands that are duplicated successfully is about 90\%--100\%.  If
    this fraction varies across different $X$'s, the relative frequencies,
    \eqref{F}, may change, which is a source of errors that we do not
    consider in this work.

\subsection{Sequencing (Reading)}

    Reading DNA is again done on silicon-based chips.  Every chip contains
    millions of readers working in parallel.  Each reader processes about 500
    nucleotides per second.  In general, the error rate is about 5\% per
    nucleotide \cite{LvB24}.

    Apart from the per-nucleotide errors, readers read the strands out of
    order, meaning that only permutation-invariant statistics can carry
    information, such as the multiplicity of each string.  Moreover, the DNA
    pool is sometimes amplified before reading.  Therefore, we are not
    counting the original strands out of synthesizers, but sampling from the
    distribution \eqref{F} (or its distortion due to amplification bias).

\subsection{Mathematical Model}

    We now state a mathematical model for the remainder of the paper to work
    with.  We assume that the synthesizer produces $M$ error-free strands,
    each $\ell$ nucleotides long.  We denote by $F(X)$ the relative frequency
    of a string $X \in \{\A, \C, \G, \T\}^\ell$ as in \eqref{F}.

    We turn on the sequencer for an amount of time that $N$ strands can be
    read.  In the noiseless case, the sequencer reports
    \[ X_1, X_2, \dotsc, X_N \in \{\A, \C, \G, \T\}^\ell, \]
    where each $X_n$ is independently sampled from the distribution \eqref{F}
    with replacement.  In the noisy case, the sequencer is affected by the
    quaternary symmetric channel $V\colon \{\A, \C, \G, \T\} \to \{\A, \C,
    \G, \T\}$ and reports
    \[ Y_1, Y_2, \dotsc, Y_N \in \{\A, \C, \G, \T\}^\ell, \]
    instead of $X_1, X_2, \dotsc, X_N$.

    We will let $M$ and $\ell$ approach infinity such that $\ell / \log_4 M$,
    usually denoted by $\beta$, remains constant.  This work is in the
    short-strand regime $M > 4^\ell$, i.e., $\beta < 1$.  The ratio $N/M$ is
    called the \emph{coverage depth}, and is fixed in this and related works,
    meaning that $N$ also grows along with $M$ and $\ell$.

    The reason that we can assume \emph{with-replacement} sampling is as
    follows.  Because amplification is essentially free of cost and free of
    errors, we actually need two parameters, $M_\flat$ and $M_\sharp$, to
    denote the numbers of strands before and after amplification,
    respectively.  Now, $M_\sharp$ controls the total number of strands that
    can be sampled with or without replacement.  But the distinction does not
    matter if $M_\sharp$ is large enough, so we choose ``with replacement''
    as it is easier to analyze.
    
    Next let us comment on the coverage depth $N/M$.  We have argued that
    $M_\sharp$ can be made arbitrarily large pretty easily so it is not
    appropriate to call $N/M_\sharp$ the coverage depth and fix it to be a
    constant.  We also notice that $M_\flat$ controls the denominators of the
    fractions \eqref{F}.  So long as $1/M_\flat$ is negligible compared to
    the sampling error \eqref{sqrt}, it contribute very little to the noise.
    Therefore, neither $N/M_\flat$ nor $N/M_\sharp$ is a fundamental
    parameter.  So we avoid referring to $M$ or talking about the ratios
    $\ell / \log M$ and $N/M$.

    Based on the above discussion, let us restate the problem as follows.
    There is a frequency vector of the form \eqref{F} whose resolution
    $M_\flat$ is high enough that the uncertainty \eqref{sqrt} dominates the
    quantization error.  The sequencer samples $N$ independent copies from a
    very large ($M_\sharp$) pool so it is essentially a with-replacement
    sampling of $F$.  The decoder is given noiseless samples $X_1, \dotsc,
    X_N$ or noisy samples $Y_1, \dotsc, Y_N$, and wants to know which $F$ was
    used to generate them.  We are interested in the maximum number of
    different $F$'s the decoder can tell apart with high probability.  The
    parameters $\ell$ and $N$ approach infinity with $\ell / \log_4 N$ fixed
    and strictly less than $1$.

\section{The Uncoded Scheme}                              \label{sec:uncoded}

    In this section, we review Tamir, Weinberger, and Guillén i Fàbregas's
    recently proposed low-complexity scheme \cite{TWG25}.

\subsection{Sender Side}

    We start with letting $\rho \in [0, 1]$ be a free parameter.  Consider
    the following arithmetic progression
    \begin{equation}
        4^{-\ell-\rho\ell},
        3 \cdot 4^{-\ell-\rho\ell},
        5 \cdot 4^{-\ell-\rho\ell}, \dotsc,
        (2 \cdot 4^{\rho\ell} - 1) \cdot 4^{-\ell-\rho\ell}.       \label{ap}
    \end{equation}
    This progression has $4^{\rho\ell}$ terms, median $4^{-\ell}$, and common
    difference $2 \cdot 4^{-\ell-\rho\ell}$.  Now each string $X \in \{\A,
    \C, \G, \T\}^\ell$ chooses one term from the progression as its frequency
    $F(X)$.  This allows $(4^{\rho\ell})^{4^\ell}$ possible messages, which
    is equivalent to $2 \rho \ell \cdot 4^\ell$ bits per pool, or $2 \rho
    \ell$ bits per string.

    Note that not all choices of $F$ will end up being a probability measure,
    i.e., $\sum_X F(X)$ needs to be $1$.  Tamir et al.\ resolved this by
    forcing all terms of the arithmetic progression \eqref{ap} to be selected
    equally often.

\subsection{Receiver Side}

    At the receiver side, each string $X$ has an empirical count $E(X) \in
    [0, N]$.  We can round $E(X) / N$ to the nearest term in the arithmetic
    progression \eqref{ap} and let it be $\hat F(X)$.  The problem boils down
    to estimating the probability that $\hat F(X) \ne F(X)$.

    We see that $E(X)$ follows the multinomial distribution with $N$
    ``balls'' and $4^\ell$ ``bins''.  The way Tamir et al.\ handled this is
    more involved.  But we claim, and will provide a rigorous proof later,
    that it suffices to model each $E(X)$ as an independent Poisson
    distribution with mean $F(X) N$.  The largest acceptable deviation that
    still leads to $\hat F(X) = F(X)$ is $|E(X) / N - F(X)| <
    4^{-\ell-\rho\ell}$, half of the common difference of \eqref{ap}, or
    $|E(X) - F(X) N| < 4^{-\ell-\rho\ell} N$.

    Now, intuitively speaking, if one accepts the Poisson model, then its
    tail bound $\exp (-\Omega(\text{deviation}^2/\text{mean}))$ implies that
    an $\hat F(X)$ differs from $F(X)$ with probability
    \[  
        \exp \Bigl( -\Omega \Bigl(
            \frac{4^{-2\ell-2\rho\ell} N^2}{4^{-\ell} N} \Bigr) \Bigr)
        = \exp \bigl( -\Omega(4^{-\ell-2\rho\ell} N) \bigr).
    \]
    Proceeding with the union bound over all $X$, the probability that $\hat
    F(X) \ne F(X)$ occurs for at least one $X$ is at most $4^\ell
    \exp(-\Omega(4^{-\ell-2\rho\ell} N))$.  This quantity is about $1$ when
    $\ell + 2\rho\ell = \log_4 N$.  Hence, the maximum number of bits each
    string can carry is $2\rho\ell = \log_4 N - \ell$, as claimed in
    Theorem~\ref{thm:noiseless}.

    At this point, readers might wonder why the previous work \cite{TWG25}
    achieved optimality only when $N < 4^{3\ell}$.  The reason is that the
    free parameter $\rho$ is assumed to be $\le 1$, and so the actual number
    of bits each string carries is $2\ell \cdot \min(1, \rho) = \min(2\ell,
    \log_4 N - \ell)$, wherein $2\ell$ is the bottleneck when $N >
    4^{3\ell}$.  The root of this assumption is that they want every term of
    the arithmetic progression \eqref{ap} to be selected by at least one
    $F(X)$, which forces $4^{\rho\ell} \le 4^\ell$.

    In the next section, we lift the restriction that each term is selected
    by at least one $F(X)$ and provide a Poisson-based rigorous analysis that
    applies to all $\varrho > 0$.

\section{Frequency Coding without Noise}                \label{sec:noiseless}

    In this section we prove Theorem~\ref{thm:noiseless} using a coded
    scheme.  Instead of parameterizing \eqref{ap} with $\rho$, we fix a small
    $\varepsilon > 0$ and let $\varrho$ be the solution to $2 (\varrho +
    \varepsilon) \ell + \ell = \log_4 N$.

\subsection{Sender Side}

    Prepare a Reed--Solomon (RS) code over the finite field of size
    $4^{\varrho\ell}$.  The code length should be $4^\ell$ and the code rate
    is $1 - \varepsilon$.  Note that the code length should not exceed the
    field size, so for the case $\varrho < 1$, we repeat the full-length RS
    code $4^{\ell-\varrho\ell}$ times and protect the RS codewords using an
    outer RS code with rate $1 - \varepsilon$.  We still call the resulting
    code a RS code for brevity.

    For the quantization levels, use the arithmetic progression
    \begin{equation}
        4^{-\ell-\varrho\ell},
        2 \cdot 4^{-\ell-\varrho\ell},
        3 \cdot 4^{-\ell-\varrho\ell}, \dotsc
        4^{\varrho\ell} \cdot 4^{-\ell-\varrho\ell}               \label{ap2}
    \end{equation}
    This progression has $4^{\varrho\ell}$ terms (the same as the finite
    field), median about $4^{-\ell}/2$, and common difference
    $4^{-\ell-\varrho\ell}$.  Naturally, each string $X$ chooses a frequency
    from \eqref{ap2} to represent one symbol of the RS codeword.

    Accordingly, each RS codeword, and hence each DNA pool, carries $2 (1 -
    \varepsilon) \varrho \ell \cdot 4^\ell$ bits of information (or $2 (1 -
    \varepsilon)^2 \varrho \ell \cdot 4^\ell$ if the outer code is used).
    This is equivalent to $2 (1 - \varepsilon) \varrho \ell$ bits per string
    (or $2 (1 - \varepsilon)^2 \varrho \ell$ if the outer code is used).
    This matches the goal claimed in Theorem~\ref{thm:noiseless} after
    setting $\varepsilon \to 0$.

    Note that $\sum_X F(X)$ is a sum of $4^\ell$ numbers each less than or
    equal to the last term of \eqref{ap2}, $4^{-\ell}$.  Therefore the
    function $F$ is almost never a probability measure.  To fix this, we ask
    $X = \A\A \dotsm \A\A$ to carry the remaining frequency, i.e., force
    $F(\A^\ell)$ to be $1 - \sum_{X\ne\A^\ell} F(X)$.  This means that the
    symbol carried by $\A^\ell$ is almost always wrong, but that is within
    the error-correcting capability of the RS code.

\subsection{Receiver Side---Resampling}

    At the receiver side, each string $X$ has an empirical count $E(X) \in
    [0, N]$.  As claimed before, we want to model $E(X)$ as a Poisson
    distribution with mean $F(X) N$.  This is made mathematically rigorous by
    the following coupling argument.

    \begin{lemma}                                         \label{lem:Poisson}
        Suppose that $E$ follows a multinomial distribution with $N$ balls
        and $4^\ell$ bins, where the bin probabilities are given by the
        frequency vector $F$.  There exist random variables $\bar E(X)$ that
        follow independent Poisson distributions with mean $F(X) N / 2$, and
        such that the probability that $\bar E(X) > E(X)$ for any $X$ is
        negligibly small.
    \end{lemma}

    \begin{IEEEproof}
        Consider the Poisson point process $P \subset [0, 1] \times [0,
        \infty)$ with a rate of one point per unit area.  We want to show
        that both $E$ and $\bar E$ can be generated from $P$ to facilitate
        comparison.  To do so, let $f\colon [0, 1] \to \{\A, \C, \G,
        \T\}^\ell$ be a lookup function that maps $x$-coordinates to strings
        and respects the frequency, i.e., $|f^{-1}(X)| = F(X)$.

        For $E$, we sort points $(x_i, y_i)$ in $P$ by their $y$-coordinates
        and pick the first $N$ points.  We then generate the statistics $E$
        by counting how many of the $N$ points fall into each $f^{-1}(X)$.
        It is easy to see that $E$ follows the multinomial distribution.

        For $\bar E$, we collect all points $(x_i, y_i)$ in $P$ below the
        line $y = N/2$.  We then generate the statistics $\bar E$ by counting
        how many of these points fall into each $f^{-1}(X)$.  It is easy to
        see that the variables $\bar E(X)$ follow independent Poisson
        distributions with means $F(X) N / 2$.  This is the first property we
        want $\bar E$ to have.

        Now, note that $\bar E$ is roughly half as greedy as $E$.  This is to
        ensure that the $N$th point $(x_N, y_N)$ of $P$ is above the line $y
        = N/2$ most of the time.  When this is the case, all points of $P$
        that contribute to $\bar E$ also contribute to $E$.  Hence $\bar E(X)
        \le E(X)$ for all $X$.  The probability of the opposite event, that
        $y_N < N/2$, is the probability that a Poisson random variable with
        mean $N/2$ is at least $N$, which is $\exp(-\Omega(N))$ by Poisson's
        tail bound.
    \end{IEEEproof}

    The quantity $\exp(-\Omega(N))$ will be counted toward the overall error
    probability, so we may safely assume that $\bar E(X) \le E(X)$ for all
    $X$ for the rest of the proof of Theorem~\ref{thm:noiseless}.  Now that
    we have the existence of $\bar E$, we next want to generate $\bar E$ from
    $E$ constructively.

    \begin{lemma}
        There is a low-complexity algorithm that turns $E$ into $\bar E$
        so that all $\bar E(X)$ are mutually independent Poisson with mean
        $F(X) N / 2$ up to total variation distance $\exp(-\Omega(N))$.
    \end{lemma}

    \begin{IEEEproof}
        We borrow the context of the previous proof.  The statistics $E$
        includes the first $N$ points of $P$ and the targeted statistics
        $\bar E$ includes all points below the line $y = N/2$.  So a
        straightforward blueprint is to generate $(x_i, y_i)$ for $i = 1, 2,
        \dotsc, N$, and stop when $y_i > N/2$.

        More precisely, the first point $(x_1, y_1)$ is such that $f(x_1)$ is
        a random string sampled from $E$ and $y_1$ is an exponential random
        variable with mean $1$.  We subtract $f(x_1)$ from $E$ and proceed to
        the next point.  The second point $(x_2, y_2)$ is such that $f(x_2)$
        is a random string sampled from the modified $E$, and $y_2 - y_1$ is
        another independent exponential random variable.  We repeat this
        process until $y_i > N/2$.

        The cost of doing what is above is updating and sampling $E$ $N$
        times and sampling an exponential random variable $N$ times.
    \end{IEEEproof}

    \begin{remark}
        An alternative approach for the lemma above is ``binary search''.
        More precisely, we can break $E$ into two parts $E'$ and $E''$ by
        generating, for each string $X$, $E'(X)$ and $E''(X)$ using
        independent binomial distributions with $E(X)$ balls and two bins of
        equal probability.  Denote by $|E'|$ the total number of balls in
        $E'$.  We then generate $y_{|E'|}$ and check if it is $> N/2$.  If
        so, it means that $E'$ contains too many points and we need to divide
        it further.  Conversely, if $y_{|E'|} < N/2$, it means that we need
        to grab more points from $E''$, as $E'$ alone is not enough to reach
        the line $y = N/2$.  By applying the same procedure recursively on
        $E'$ or $E''$, we halve the size of the statistics at each step and
        will terminate in $O(\log N)$ steps.  Each step needs $O(4^\ell)$
        time to split the statistics.
    \end{remark}

    To conclude this subsection, we demonstrate how to generate variables
    $\bar E(X)$ that follow independent Poisson distributions with mean $F(X)
    N / 2$ up to a negligible failure probability $\exp(-\Omega(N))$.  We
    then use $\bar E$ to guess $F$ in the next subsection.

\subsection{Receiver Side---Decoding}

    Following the same logic as in Section~\ref{sec:uncoded}, we now round $2
    \bar E(X) / N$ to the nearest term in the arithmetic progression
    \eqref{ap2} and let it be $\hat F(X)$.  It remains to check the
    probability of $\hat F(X) \ne F(X)$ for each $X \ne \A^\ell$.

    Each equality $\hat F(X) = F(X)$ fails to hold if $|2 \bar E(X) / N -
    F(X)| \ge 4^{-\ell-\varrho\ell} / 2$.  This is equivalent to controlling
    the deviation of a Poisson random variable with mean $F(X) N / 2$ within
    $4^{-\ell-\varrho\ell} N / 4$.  For this, we borrow the following tail
    bound.

    \begin{lemma}
        Let $Z$ be a Poisson random variable with mean $\lambda$.  For any $t
        > 0$,
        \[
            \PP \{|Z - \lambda| \ge t\} \le
            2 \exp \Bigl( \frac{-t^2}{2 (\lambda + t)} \Bigr).
        \]
    \end{lemma}

    Using the lemma with $\lambda = F(X) N / 2$ and $t =
    4^{-\ell-\varrho\ell} N / 4$, we immediately see that
    \[ 
        \PP \{\hat F(X) \ne F(X)\} \le 2 \exp \Bigl( -\Omega \Bigl(
            \frac{4^{-2\ell-2\varrho\ell} N^2}{4^{-\ell} N} \Bigr) \Bigr),
    \]
    where the denominator uses the last term of \eqref{ap2} to upper bound
    $F(X)$.  The exponent simplifies to $-\Omega(4^{-\ell-2\varrho\ell} N)$.
    Recall that we choose $\varrho = (\log_4 N - \ell) / 2\ell -
    \varepsilon$, so the probability of each $\hat F(X) \ne F(X)$ is at most
    $\exp(-\Omega(4^{2\varepsilon\ell}))$.

    With per-symbol error probability this low, the RS decoder fails to
    recover the message with probability at most
    $\exp(-\Omega(4^{(1+2\varepsilon)\ell}))$ by bounding the probability
    that the fraction of bad symbols exceeds $\varepsilon/3$ using the
    standard Chernoff bound.  A similar computation can be performed if the
    outer code is used and the resulting error probability is similar.  Note
    that we need independence across different $X$'s for this step, which is
    why we spent considerable effort on the Poisson point process.

    This concludes the proof of Theorem~\ref{thm:noiseless}.  We have
    designed a low-complexity scheme that achieves $\log_4 N - \ell$ bits per
    string for the entire short-strand regime.

\section{Frequency Coding with Noise}                       \label{sec:noise}

    In this section, we prove Theorem~\ref{thm:noisy}.  The main technique is
    to move to the Fourier domain because the channel (recall that it is
    \[
        V(y|x) = \bca{
            1 - p, & if $y = x$, \\
            p/3, & if $y \ne x$.
        }
    \]
    and is symmetric) is diagonalized by the Fourier transform.  This allows
    us to design exactly the right quantization levels for every Fourier
    mode.

    More precisely, $V$ in \eqref{channel} can be diagonalized as
    \[
        H D H =
        \frac12 \sma{
            1 & 1 & 1 & 1\\
            1 & 1 & -1 & -1\\
            1 & -1 & 1 & -1\\
            1 & -1 & -1 & 1
        } \cdot
        \sma{
            1 & & & \\
            & \delta & & \\
            & & \delta & \\
            & & & \delta
        } \cdot
        \frac12 \sma{
            1 & 1 & 1 & 1\\
            1 & 1 & -1 & -1\\
            1 & -1 & 1 & -1\\
            1 & -1 & -1 & 1
        }
    \]
    for $\delta \coloneqq 1 - 4p/3 > 0$.  This implies that $W =
    V^{\otimes\ell}$ can be diagonalized as $H^{\otimes\ell} D^{\otimes\ell}
    H^{\otimes\ell}$.  In particular, if we consider the representation
    vector $G \coloneqq H^{\otimes\ell} F$ in the Fourier domain, then the
    Fourier representation of $WF$ is $D^{\otimes\ell} G$, meaning that it is
    now a coordinate-wise contraction.

    We now describe our scheme.

\subsection{Sender Side}

    The sender will put information on the Fourier vector $G$.  But the
    choice is subject to three constraints.  First, the frequency vector $F$
    obtained by $H^{\otimes\ell} G$ must be nonnegative.  Second, $F$ must be
    a probability measure, so the leading coefficient of $G$ is fixed to be
    $2^{-\ell}$.  Third, the quantization levels of $G$ must be adaptive to
    the channel attenuation $\delta$.  We claim, and will verify on the
    receiver side, that after moving to the Fourier domain, all but a small
    fraction of coordinates have count noise at most $\theta \coloneqq \tau
    \sqrt{N / 4^\ell}$, where $\tau \to \infty$ grows very slowly compared to
    $\ell$ and $N$.  Therefore, if a Fourier coefficient is to be multiplied
    by $N \delta^k$, then a spacing of order $\theta / (N \delta^k)$ is
    enough for correct rounding.

    For nonnegativity, let us limit each non-leading Fourier coefficient to
    $[0, 8^{-\ell}]$.  Each mass of $F$ is a sum of $4^\ell$ Fourier
    coefficients divided by the normalizer of the Hadamard matrix $2^{\ell}$.
    This ensures that each $F(X)$ deviates from the uniform distribution's
    mass $4^{-\ell}$ by at most $4^\ell \cdot 8^{-\ell} / 2^{\ell} =
    4^{-\ell}$, implying that they are all nonnegative.

    Next, we analyze the effect of multiplying the Fourier coefficients by
    $\delta$.  With the above threshold, a degree-$k$ coefficient (a
    coefficient that is multiplied by $N \delta^k$) can provide about
    $8^{-\ell} N \delta^k / \theta = \delta^k \sqrt N / (4^\ell \tau)$
    choices to encode information and can carry
    \begin{equation}
        \log_4 N + k \log_2 \delta - 2\ell - \log_2 \tau         \label{bits}
    \end{equation}
    bits if this quantity is positive, and we do not lose any bit if it is
    negative.
    
    The total number of bits all coefficients can carry is thus at least
    \begin{multline*}
        \sum_{k=1}^\ell \binom{\ell}{k} 3^k \cdot \eqref{bits} \approx
        \sum_{k=0}^\ell \binom{\ell}{k} 3^k \cdot \eqref{bits}
        \\ = 4^\ell (\log_4 N - 2\ell - \log_2 \tau)
        + 3 \ell 4^{\ell-1} \log_2 \delta.
    \end{multline*}
    The second term is exactly $\log_2 \det W$ because the eigenvalues of $W$
    are $\delta^k$ with multiplicity $\binom{\ell}{k}3^k$.  The rate we lose
    for not utilizing the $k = 0$ coefficient (to ensure that $F$ is a
    probability measure) is $o(1)$ of the other terms.  As a result, the
    number of bits per pool is $(1 \pm o(1))$-fold within $4^\ell (\log_4 N -
    2\ell - \log_2 \tau) + \log_2 \det W$.  Since $\tau$ grows slowly, the
    $\log_2 \tau$ term is negligible and the remainder matches the
    information rate stated in Theorem~\ref{thm:noisy}.

\subsection{Receiver Side}

    Let $E(X)$ be the number of times $Y_i = X$.  Let $\bar E$ be the iid
    Poisson random variables obtained by the recipe of
    Section~\ref{sec:noiseless}-B: $\bar E(X) \sim \operatorname{Poisson}(N
    Q(X) / 2)$, where $Q \coloneqq W F$ is the post-channel frequency vector.
    Equivalently, $2 \bar E \eqqcolon N Q + Z$.  This $Z$ represents the
    centered sampling noise; each $Z(X)$ has mean zero and variance $2 N
    Q(X)$.  So the expected norm squared is $\EE \|Z\|_2^2 = \sum_X 2 N Q(X)
    = 2 N$.  Since $\|Z\|_2^2$ is a sum of independent random variables, it
    concentrates around its mean $2 N$ with high probability (see the next
    subsection for quantitative bounds).

    We now switch to the Fourier domain,
    \[
        H^{\otimes\ell} (2 \bar E)
        = N H^{\otimes\ell} Q + H^{\otimes\ell} Z
        = N D^{\otimes\ell} G + H^{\otimes\ell} Z.
    \]
    By Parseval, $\|H^{\otimes\ell} Z\|_2$ is still $\|Z\|_2$.  Thus the
    Fourier transform does not reduce the sampling noise, nor does it amplify
    it.  It merely redistributes the noise across different Fourier
    coefficients.

    This redistribution cannot damage too many Fourier coordinates at once.
    Indeed, let $Y \coloneqq H^{\otimes\ell} Z$.  If $\|Y\|_2^2 \le 4 N$,
    which is likely to happen by concentration, then for any threshold
    $\theta > 0$,
    \[
        \bigl| \bigl\{X: |Y(X)| > \theta \bigr\} \bigr|
        \le \frac{4 N}{\theta^2}.
    \]
    Since there are $4^\ell$ Fourier coordinates, choosing $\theta =
    \tau\sqrt{N / 4^\ell}$ gives at most $(4/\tau^2) 4^\ell$ bad coordinates.
    Note that these are not iid errors, so we model them as adversarial
    errors.  That said, they can still be countered by a capacity-achieving
    code designed for BSC$(O(1/\tau^2))$ via a random shuffling argument.
    The code rate we lose is negligible if $\tau \to \infty$ (see the next
    subsection for details).

    The receiver therefore treats each Fourier coordinate separately.  In
    degree $k$, the channel multiplies the corresponding coefficient of $G$
    by $N \delta^k$. Hence a spacing of order $\theta / (N \delta^k)$ in that
    coordinate is enough for correct rounding, except on the coordinates
    where $|Y(X)| > \theta$.  The sender-side quantization was chosen with
    exactly this spacing, and the extra $\log_2 \tau$ bits per coordinate are
    negligible because $\tau$ grows slowly.  This means that the receiver can
    decode successfully with high probability, and finishes the proof of
    Theorem~\ref{thm:noisy}.

\subsection{Error Probability Analysis}

    Throughout the proof in this section, there are three sources that may
    cause decoding failure.  The first source is the rare event that $\bar E$
    is not less than $E$ entry-wise.  By Lemma~\ref{lem:Poisson}, this
    contributes $\exp(-\Omega(N))$ and is minor compared to the other
    sources.

    The second source is that $\|Z\|_2^2$ may be too large.  Notice that the
    situation here is somewhat classical in concentration literature but not
    entirely.  On the one hand, $Z$ is a vector of independent random
    variables, but the tail of Poisson is not sub-Gaussian, only
    sub-exponential.  On the other hand, we are not interested in the sum of
    random variables, but their squares.  Luckily, a black-box bound is
    available in the literature.

    \begin{proposition}
        Let $Z_1, \dotsc, Z_n$ be independent random variables satisfying
        $\EE Z_i = 0$ and $\EE \exp(|Z_i| / M) \le 2$ for every $i$.  Let $A$
        be a symmetric $n \times n$ matrix.  Then, for every $t > 0$,
        \begin{multline*}
            \PP\biggl\{ \biggl|
                \sum_{i,j} a_{ij} Z_i Z_j - \text{expectation}
            \biggr| \ge t \biggr\}
            \\ \le 2 \exp\biggl( -\Omega \biggl( \min \biggl\{
                \frac{t^2}{M^4 \|A\|_\mathrm{HS}^2},
                \sqrt{\frac{t}{M^2 \|A\|_\mathrm{op}}}
            \biggr\} \biggr) \biggr).
        \end{multline*}
        For a proof, see \cite[Proposition~1]{GSS21} with $\alpha = 1$.
    \end{proposition}

    To use the proposition, let $A$ be the identity matrix of size $4^\ell$.
    Then $\|A\|_\mathrm{HS}^2 = 4^\ell$ and $\|A\|_\mathrm{op} = 1$.  To
    determine $M$, note that $Q(X) \le 2/4^\ell$ for all $X$ because we know
    that $F(X)$ does not deviate from the mean $4^{-\ell}$ by more than
    $4^{-\ell}$ when we were making sure that $F$ is nonnegative.  Thus each
    $Z(X)$ has variance $O(N/4^\ell)$.  Dividing $Z(X)$ by its standard
    deviation $M \coloneqq O(\sqrt{N/4^\ell})$ puts the fluctuation on a
    constant scale.  With a standard calculation, which we omit, this $M$
    makes $\EE \exp(|Z_i| / M) \le 2$.

    Now let $t = 2 N$ to obtain
    \begin{multline*}
        \PP \{ \|Z\|_2^2 > 4 N \} \\
        \le 2 \exp\bigl( -\Omega( \min\{4^\ell, \sqrt{4^\ell}\}) \bigr)
        = \exp\bigl( -\Omega(2^\ell) \bigr).
    \end{multline*}
    While this is not exponential in $4^\ell$ or $N$, it still suits our
    need.

    The third source of error is that the capacity-achieving code needed to
    protect the Fourier coefficients fails to decode.  In fact, since
    \eqref{bits} is not a constant, it is not a priori clear how to construct
    the code.  To clarify this, we prepare a family of binary codes that
    achieve the capacity of BSC$(C/\tau^2)$, where the constant $C$ will be
    specified later.  For each Fourier coefficient, we assign it randomly to
    $\max(0, \lfloor \eqref{bits} \rfloor)$ bits of a codeword.  The
    assignment needs to be random to turn adversarial errors into random
    ones.  This assignment also loses $O(1)$ bits per coefficient, which will
    be absorbed by $(1 - o(1))$.  Next, we note that good capacity-achieving
    codes enjoy exponential error probability, so the error probability
    contributed by this part is $\exp(-\Omega(4^\ell))$.

    It remains to specify the constant $C$.  Recall that different
    coefficients correspond to different amounts of bits, so a $4/\tau^2$
    fraction of bad coefficients does not directly translate to the same
    fraction of bad bits.  However, observe that $r \coloneqq \log_4 R - \ell
    + 3\ell \log_{16} \delta$ is in fact a constant multiple of $\ell$ ($\ell
    / \log_4 R$ is constant because $\ell / \log_4 N$ is fixed, and $\delta$
    is a constant because $p$ is fixed).  Also observe that the leading term
    of \eqref{bits}, $\log_4 N$, is a constant multiple of $\ell$.  This
    implies that each coefficient corresponds to at most $\log_4 N$ bits and
    on average $r$ bits, so the number of errors only increases by a factor
    of $\log_4(N)/r$.  Hence choosing $C$ to be slightly larger than
    $\log_4(N)/r$ suffices.

    In conclusion, the error probability of Theorem~\ref{thm:noisy} is the
    sum of three sources, $\exp(-\Omega(N))$, $\exp(-\Omega(2^\ell))$, and
    $\exp(-\Omega(4^\ell))$, which is $\exp(-\Omega(2^\ell))$.

\section{Conclusion}

    In this paper we discuss how to use frequency vectors to store data.  In
    the noiseless case, we design a low-complexity scheme that completes the
    short-strand regime left open by the previous work \cite{TWG25}.  In the
    noisy case, we use the Fourier transform to switch to a basis in which
    the noise is easily predictable, and design a low-complexity scheme that
    is a small distance away from the capacity characterization by
    \cite{GSW26}.

\bibliographystyle{IEEEtran}
\bibliography{NoisyFrequency-26}

% Generated by IEEEtran.bst, version: 1.14 (2015/08/26)
\begin{thebibliography}{10}
\providecommand{\url}[1]{#1}
\csname url@samestyle\endcsname
\providecommand{\newblock}{\relax}
\providecommand{\bibinfo}[2]{#2}
\providecommand{\BIBentrySTDinterwordspacing}{\spaceskip=0pt\relax}
\providecommand{\BIBentryALTinterwordstretchfactor}{4}
\providecommand{\BIBentryALTinterwordspacing}{\spaceskip=\fontdimen2\font plus
\BIBentryALTinterwordstretchfactor\fontdimen3\font minus
  \fontdimen4\font\relax}
\providecommand{\BIBforeignlanguage}[2]{{%
\expandafter\ifx\csname l@#1\endcsname\relax
\typeout{** WARNING: IEEEtran.bst: No hyphenation pattern has been}%
\typeout{** loaded for the language `#1'. Using the pattern for}%
\typeout{** the default language instead.}%
\else
\language=\csname l@#1\endcsname
\fi
#2}}
\providecommand{\BIBdecl}{\relax}
\BIBdecl

\bibitem{MiP24}
O.~Milenkovic and C.~Pan, ``{{DNA-Based Data Storage Systems}}: {{A Review}} of
  {{Implementations}} and {{Code Constructions}},'' \emph{IEEE Transactions on
  Communications}, vol.~72, no.~7, pp. 3803--3828, Jul. 2024.

\bibitem{SKS24}
O.~Sabary, H.~M. Kiah, P.~H. Siegel, and E.~Yaakobi, ``Survey for a {{Decade}}
  of {{Coding}} for {{DNA Storage}},'' \emph{IEEE Transactions on Molecular,
  Biological, and Multi-Scale Communications}, vol.~10, no.~2, pp. 253--271,
  Jun. 2024.

\bibitem{XLC24}
\BIBentryALTinterwordspacing
L.~Xiang, Q.~Liu, S.~Chen, K.~Yan, W.~Wu, and K.~Yang, ``A {{Tutorial}} on
  {{Coding Methods}} for {{DNA-Based Molecular Communications}} and
  {{Storage}},'' \emph{IEEE Internet of Things Journal}, vol.~11, no.~7, pp.
  11\,825--11\,847, Apr. 2024. [Online]. Available:
  \url{https://ieeexplore.ieee.org/document/10319319/}
\BIBentrySTDinterwordspacing

\bibitem{BMY22}
\BIBentryALTinterwordspacing
D.~{Bar-Lev}, S.~Marcovich, E.~Yaakobi, and Y.~Yehezkeally, ``Adversarial
  {{Torn-paper Codes}},'' in \emph{2022 {{IEEE International Symposium}} on
  {{Information Theory}} ({{ISIT}})}.\hskip 1em plus 0.5em minus 0.4em\relax
  Espoo, Finland: IEEE, Jun. 2022, pp. 2934--2939. [Online]. Available:
  \url{https://ieeexplore.ieee.org/document/9834766/}
\BIBentrySTDinterwordspacing

\bibitem{RVS21}
\BIBentryALTinterwordspacing
A.~N. Ravi, A.~Vahid, and I.~Shomorony, ``Capacity of the {{Torn Paper
  Channel}} with {{Lost Pieces}},'' in \emph{2021 {{IEEE International
  Symposium}} on {{Information Theory}} ({{ISIT}})}.\hskip 1em plus 0.5em minus
  0.4em\relax Melbourne, Australia: IEEE, Jul. 2021, pp. 1937--1942. [Online].
  Available: \url{https://ieeexplore.ieee.org/document/9518272/}
\BIBentrySTDinterwordspacing

\bibitem{ShV21}
\BIBentryALTinterwordspacing
I.~Shomorony and A.~Vahid, ``Torn-{{Paper Coding}},'' \emph{IEEE Transactions
  on Information Theory}, vol.~67, no.~12, pp. 7904--7913, Dec. 2021. [Online].
  Available: \url{https://ieeexplore.ieee.org/document/9576724/}
\BIBentrySTDinterwordspacing

\bibitem{RGC26}
\BIBentryALTinterwordspacing
R.~Rebimbas, I.~Gl{\'o}ria, J.~Cheg{\~a}o, M.~{Al-Rawi}, A.~Mousakhani~Ganjeh,
  and J.~A. Saraiva, ``{{DNA}} as a data storage medium,'' \emph{Journal of
  Biotechnology}, vol. 414, pp. 19--35, Jun. 2026. [Online]. Available:
  \url{https://linkinghub.elsevier.com/retrieve/pii/S0168165626000702}
\BIBentrySTDinterwordspacing

\bibitem{WeM22}
N.~Weinberger and N.~Merhav, ``The {{DNA Storage Channel}}: {{Capacity}} and
  {{Error Probability Bounds}},'' \emph{IEEE Transactions on Information
  Theory}, vol.~68, no.~9, pp. 5657--5700, Sep. 2022.

\bibitem{LSW23}
A.~Lenz, P.~H. Siegel, A.~{Wachter-Zeh}, and E.~Yaakobi, ``The {{Noisy Drawing
  Channel}}: {{Reliable Data Storage}} in {{DNA Sequences}},'' \emph{IEEE
  Transactions on Information Theory}, vol.~69, no.~5, pp. 2757--2778, May
  2023.

\bibitem{WaG25}
H.-P. Wang and V.~Guruswami, ``Geno-weaving: A framework for low-complexity
  capacity-achieving dna data storage,'' \emph{IEEE Journal on Selected Areas
  in Information Theory}, vol.~6, pp. 383--393, 2025.

\bibitem{ShH22}
I.~Shomorony and R.~Heckel, ``Information-{{Theoretic Foundations}} of {{DNA
  Data Storage}},'' \emph{Foundations and Trends\textregistered{} in
  Communications and Information Theory}, vol.~19, no.~1, pp. 1--106, 2022.

\bibitem{GSW25}
Y.~Gerzon, I.~Shomorony, and N.~Weinberger, ``Capacity of frequency-based
  channels: Encoding information in molecular concentrations,'' \emph{IEEE
  Transactions on Information Theory}, vol.~71, no.~8, pp. 5788--5808, Aug
  2025.

\bibitem{TWG25}
R.~Tamir, N.~Weinberger, and A.~G. i~Fàbregas, ``Dna storage in the short
  molecule regime,'' 2025.

\bibitem{GSW26}
Y.~Gerzon, I.~Shomorony, and N.~Weinberger, ``On the capacity of noisy
  frequency-based channels,'' 2026.

\bibitem{FKK25}
S.~M. Forget, M.~J. Krawczyk, A.~M. Knight, C.~Ching, R.~A. Copeland,
  N.~Mahmoodi, M.~A. Mayo, J.~Nguyen, A.~Tan, M.~Miller, J.~Vroom, and S.~Lutz,
  ``Evolving a terminal deoxynucleotidyl transferase for commercial enzymatic
  {{DNA}} synthesis,'' \emph{Nucleic Acids Research}, vol.~53, no.~4, p.
  gkaf115, Feb. 2025.

\bibitem{FMS21}
S.~Filges, P.~Mouhanna, and A.~St{\aa}hlberg, ``Digital {{Quantification}} of
  {{Chemical Oligonucleotide Synthesis Errors}},'' \emph{Clinical Chemistry},
  vol.~67, no.~10, pp. 1384--1394, Oct. 2021.

\bibitem{MOS22}
Y.~Masaki, Y.~Onishi, and K.~Seio, ``Quantification of synthetic errors during
  chemical synthesis of {{DNA}} and its suppression by non-canonical
  nucleosides,'' \emph{Scientific Reports}, vol.~12, no.~1, p. 12095, Jul.
  2022.

\bibitem{YTL24}
M.~Yu, X.~Tang, Z.~Li, W.~Wang, S.~Wang, M.~Li, Q.~Yu, S.~Xie, X.~Zuo, and
  C.~Chen, ``High-throughput {{DNA}} synthesis for data storage,''
  \emph{Chemical Society Reviews}, vol.~53, no.~9, pp. 4463--4489, 2024.

\bibitem{Sai88}
R.~K. Saiki, D.~H. Gelfand, S.~Stoffel, S.~J. Scharf, R.~Higuchi, G.~T. Horn,
  K.~B. Mullis, and H.~A. Erlich, ``Primer-directed enzymatic amplification of
  dna with a thermostable dna polymerase,'' \emph{Science}, vol. 239, no. 4839,
  pp. 487--491, 1988.

\bibitem{Joh93}
K.~A. Johnson, ``Conformational coupling in dna polymerase fidelity.''
  \emph{Annual review of biochemistry}, vol.~62, pp. 685--713, 1993.

\bibitem{LvB24}
W.~{Liu-Wei}, W.~{van~der~Toorn}, P.~Bohn, M.~H{\"o}lzer, R.~P. Smyth, and
  M.~Von~Kleist, ``Sequencing accuracy and systematic errors of nanopore direct
  {{RNA}} sequencing,'' \emph{BMC Genomics}, vol.~25, no.~1, p. 528, May 2024.

\bibitem{GSS21}
\BIBentryALTinterwordspacing
F.~Götze, H.~Sambale, and A.~Sinulis, ``Concentration inequalities for
  polynomials in alpha-sub-exponential random variables,'' \emph{Electronic
  Journal of Probability}, vol.~26, no. none, Jan. 2021. [Online]. Available:
  \url{http://dx.doi.org/10.1214/21-EJP606}
\BIBentrySTDinterwordspacing

\end{thebibliography}

\end{document}